\documentclass[runningheads]{llncs}

\usepackage[utf8]{inputenc}
\usepackage[english]{babel}
\RequirePackage{amsfonts}
\RequirePackage{amssymb}
\RequirePackage{mathtools}
\usepackage[T1]{sansmath}
\SetMathAlphabet{\mathsfbf}{sans}{\sansmathencoding}{\sfdefault}{bx}{sl}

\usepackage{amsthm}
\usepackage{xspace}

\usepackage{xcolor}
\usepackage{textcomp}
\usepackage{listings}

\definecolor{darkgray}{rgb}{.4,.4,.4}

\usepackage{microtype}
\usepackage{array,tabularx}
\usepackage[pdfusetitle]{hyperref}
\usepackage[inline]{enumitem}
\usepackage{csquotes}
\usepackage{expl3}
\usepackage{ebproof}

\newcommand\hypo{\Hypo}
\newcommand\infer{\Infer}

\usepackage{subcaption}

\addto\extrasenglish{

}

\newcommand{\BNF}{\enspace \ensuremath{\Vert} \enspace} %

\DeclareMathOperator{\id}{id}

\DeclareMathOperator{\var}{var}

\newcommand{\mat}[1]{\left(\begin{smallmatrix}#1\end{smallmatrix}\right)} %
\newcommand{\zmat}{\mathbf{0}} %
\newcommand{\umat}{\mathbf{1}} %

\newcommand{\pr}{\lstinline[mathescape]}
\newcommand{\prc}{\lstinline[language=C]}

\usepackage{hyperref}
\makeatletter
\AtBeginDocument{%
	\def\doi#1{\url{https://doi.org/#1}}}
\makeatother

\ebproofnewstyle{small}{
	separation = 1em, rule margin = .5ex,
}

\newcommand*{\eg}{e.g.\@\xspace}
\newcommand*{\cf}{cf.\@\xspace}
\newcommand*{\ie}{i.e.\@\xspace}

\newcommand*{\Ie}{I.e.\@\xspace}

\makeatletter
\newcommand*{\etc}{%
	\@ifnextchar{.}%
	{etc}%
	{etc.\@\xspace}%
}
\makeatother

\usepackage{soul}

\newcommand*{\resp}{resp.\@\xspace}
\newcommand*{\wrt}{w.r.t.\@\xspace}

\renewcommand{\oplus}{+}
\renewcommand{\otimes}{\times}

\makeatletter
\renewcommand{\boxed}[1]{\text{\fboxsep=.2em\fbox{\m@th$\displaystyle#1$}}}
\makeatother

\newcommand{\vdashJK}{\vdash_{\textsc{jk}}} \allowdisplaybreaks[1] %

\makeatletter%
\begin{document}

\title{An extended and more practical mwp flow analysis\thanks{
		This material is based upon research supported by the \href{https://face-foundation.org/higher-education/thomas-jefferson-fund/}{Thomas Jefferson Fund} of the Embassy of France in the United States and the
		\href{https://face-foundation.org/}{FACE Foundation}. Thomas Rubiano and Thomas Seiller are also supported by the Île-de-France region through the DIM RFSI project "CoHOp".} 
}
\author{
	Clément Aubert\inst{1}\orcidID{0000-0001-6346-3043} \and
	Thomas Rubiano\inst{2} \and
	Neea Rusch\inst{1} \and
	Thomas Seiller\inst{2,3}\orcidID{0000-0001-6313-0898}
}
\authorrunning{C. Aubert et al.}

\institute{
	School of Computer and Cyber Sciences, Augusta University \and
	LIPN – UMR 7030 Université Sorbonne Paris Nord \and
	CNRS
}

\maketitle %

\begin{abstract}
	We improve and refine a method for certifying that the values' sizes computed by an imperative program will be bounded by polynomials in the program's inputs' sizes.
	Our work \enquote{tames} the non-determinism of the original analysis, and offers an innovative way of completing the analysis when a non-polynomial growth is found.
	We furthermore enrich the analyzed language by adding function definitions and calls, allowing to compose the analysis of different libraries and offering generally more modularity.
	The implementation of our improved method, discussed in a tool
	paper~\cite{Aubert2021f}, also required to reason about the efficiency of some of the needed operations on the matrices produced by the analysis.
	It is our hope that this work will enable and facilitate static analysis of source code to guarantee its correctness with respect to resource usages.
\end{abstract}


\section{Introduction}
\label{sec:intro}

This work takes a step further in the implementation of static analysis methods inspired from work in implicit computational theory~\cite{DalLago2012a}, and more particularly the series of work from the so-called flow analysis of the \enquote{Copenhagen school}, notably Neil Jones, Lars Kristiansen, and Amir Ben-Amram. The \emph{Copenhagen school approach to implicit computational complexity} initiates in the celebrated \emph{size-change principle} of Ben Amram~\cite{Ben-Amram2008} to characterise termination of programs, and evolved in more precise techniques to capture resource usage and more particularly growth rate between variables' sizes.
The overall flow analysis approach is related in spirit to \emph{abstract interpretation}~\cite{Cousot1977b,Cousot1977a}; as noted by Jones~\cite{Jones1995} it bounds \emph{transitions} between states (\eg commands) instead of states. A first work implemented this technique to develop a static analysis tool detecting loop quasi-invariants~\cite{Moyen2017,Moyen2017b}.

One landmark result in this series of work is the precise and detailed analysis of the relationship between the resource requirements of a computation and the way data might flow during the computation~\cite{Jones2009}.
Thanks to a typing system resting on matrices with coefficients in the so-called mwp semi-ring, programs in a basic imperative language could be guaranteed to have variables growing at most polynomially with respect to their inputs. While this does not ensure termination, it provides a certificate that \emph{if} the program terminates, it will do so in polynomial time and using at most polynomial space.
We here continue in this direction and implement this \enquote{mwp-bounds analysis}~\cite{Jones2009} on a fragment of the \texttt{C} programming language~\cite{Aubert2021f}.

Our contribution is both of practical and theoretical interest: a \texttt{python} program analyzing \texttt{C} source code is currently being developed~\cite{Aubert2021f}, and documented at \url{https://seiller.github.io/pymwp/}.
This implementation largely benefited from the improvements discussed in the current paper and triggered the development of a modified analysis allowing for the use of more efficient algorithms to carry out the analysis.
However, our theoretical contributions can be read independently from this implementation, and answers some of the questions asked by Jones and Kristiansen~\cite[Section 1.2]{Jones2009}, notably pushing further their method.
Two of those questions are
\begin{enumerate*}
	\item Can the method be extended to richer languages?
	\item How powerful and convenient is this method?
\end{enumerate*}
We answer the first question positively, by adding treatment of function
definition and calls, and by implementing the analysis on an actual programming
language instead of a simple imperative language.
Our work suggests that to answer the second question precisely, a lot of care is needed: the analysis uses matrices in a non-deterministic way to compute the influence of variables on each others, resulting in objects growing exponentially in number.
By altering the semi-ring at the core of the original analysis, we show not only that more parsimonious methods can be used, but also that the mathematical machinery can be substituted.
While our alternative approach essentially carries out the same analysis, we improved the implementability (and in fact implemented it already), re-usability and efficiency of the techniques while illustrating that the general method could be adapted easily for different types of analysis.

\subsection{Complexity, resource growth and implementations: a brief tour}

Our approach is conceptually guided by implicit computational complexity, that generally focus on restricting \eg recursion~\cite{Bellantoni1992,Leivant1993} or type systems~\cite{Baillot2004,Lafont2004} to insure that a programming language captures a particular complexity class, or perform amortized resource analysis~\cite{Hofmann2015b}.
The particular domain concerned here, \emph{data-flow analysis}, more specifically focuses on measuring or restricting loops in imperative programs~\cite{Jones2009,Kristiansen2004,Niggl2006} and was implemented on \eg low level assembly-like programs~\cite{Moyen2009}.

However, our work is probably best compared with approaches coming from other communities sharing the same goal of finding worst case resource consumption.
Complexity analyzers of different languages, such as SPEED~\cite{Gulwani2009} for \texttt{C}++,
COSTA~\cite{Costa2007} for \texttt{Java} bytecode, RaML for \texttt{OCaml}~\cite{Lichtman2017} or
Cerco~\cite{Cerco2014} for \texttt{C} all attempts to generate (possibly
certified) cost annotation on (subsets of) programming languages in use.

\subsection{Contribution: a different take on implementing the theory}
\label{ssec:contrib}

We would like to argue that the \enquote{mwp} approach we are extending and making more practical is different from the previously mentioned implementations in four respects: \begin{enumerate*}
	\item it is focused on the \emph{growth} of variables instead of focusing on their possible values,
	\item it is modular, in the sense that the internal machinery can be altered -- as we illustrate in this paper -- without the need to re-develop large chunks of the theory,
	\item it is at the same time language-independent, as it reasons abstractly on imperative languages, and easy to port, as we illustrate with our implementation~\cite{Aubert2021f},
	\item it focuses on characterizations of \enquote{chunks} of any size of the program allowing to abstract values and their encoding.
\end{enumerate*}

\section{Background: the original flow analysis}

We here quickly recall the original syntax of the imperative language, that we will extend with function call and definition in \autoref{composition}, then the original analysis by Jones and Kristiansen and its mathematical machinery.

\subsection{Language analyzed: fragments of imperative language}
\label{subsec:language}

We will be using the following imperative programming language, where variables range over \pr|R|, \pr|X|, \pr|X'|, \pr|Y|, \pr|Z| and \pr|Xi| for \pr|i| \(\in \mathbb{N}\), and need not to be declared, and given the binary operations on expression \(-\), \(+\), and \(\times\), and on booleans \(\bullet\) (such as \(\wedge\), \(\vee\), \etc), and a collection of comparison operators between expressions\(\square\). %

\begin{align*}
	\text{\pr|e|} \coloneqq & \text{\pr|X|} \BNF \text{\pr|e - e|} \BNF \text{\pr|e + e|} \BNF \text{\pr|e * e|} \tag{Expression} \\
	\text{\pr|b|} \coloneqq & \text{\pr|e $\square$ e|} \BNF \text{\pr|b $\bullet$ b|} \tag{Boolean expression}                   \\
	\text{\pr|C|} \coloneqq & \text{\pr|X = e|} \BNF \text{\pr|if b then C else C|} \BNF \text{\pr|while b do \{C\}|}             \\
	                        & \text{\pr|loop X \{C\}|} \BNF \text{\pr|C ; C|} \tag{Command}
\end{align*}

The semantics is straightforward, with \pr|loop X {C}| meaning \enquote{do \pr|C| \pr|X| times} and \pr|C;C| being used for sequentiality (\enquote{do \pr|C|, then \pr|C|}).
We generally write \enquote{program} for a series of commands composed sequentially.

\subsection{A Flow Calculus of mwp-Bounds for Complexity Analysis}
\label{ssec:a-flow-calculus}

The original paper~\cite{Jones2009} studies \emph{flows} between variables in imperative programs, that are of three types: \emph{m}aximum, \emph{w}eak polynomial and \emph{p}olynomial flows characterize the three forms of controls from one variable to another, with increasing growth rate\footnote{Note that \(0\) is also a possible type, corresponding to the absence of any dependency.}.
The programs are written in (a variation on) the language presented in \autoref{subsec:language}, and %
the bounds are represented and calculated thanks to vector and matrices whose coefficients are elements of the mwp semi-ring.

\begin{definition}[mwp semi-ring, matrix algebra]
	\label{def:mwp-matrix-alg}
	Letting \(\textsc{mwp} = \{0, m, w, p\}\) with \(0 < m < w < p\), and \(\alpha\), \(\beta\), \(\gamma\) range over \(\textsc{mwp}\), the \emph{mwp semi-ring} \((\textsc{mwp}, 0, m, +, \times)\) is defined with \(+ = \max\), \(\alpha \times \beta = \max (\alpha, \beta)\) if \(\alpha, \beta \neq 0\), and \(0\) otherwise.

	Fixing a natural number $n$, we use \(M\), \(A\), \(B\), \(C\), \(\hdots\) to denote \(n \times n\) matrices over \(\textsc{mwp}\), \(M_{ij}\) for the coefficient in the \(i\)th row and \(j\)th column of \(M\), \(\oplus\) for the component wise addition, and \(\otimes\) for the product of matrices defined in a standard way.
	The \(\zmat\)-element for the addition is \(\zmat_{ij} = 0\) for all \(i, j\), and the \(\umat\)-element for product is \(\umat_{ii} = m\), \(\umat_{ij} = 0\) if \(i \neq j\), and the resulting structure \(\mathbb{M}(\textsc{mwp})\) is a semi-ring.
	Finally, \(M^0 = \umat\), \(M^{m+1} = M \otimes M^m\) and the closure operator \(\cdot^{*}\) is defined as \(M^* = \umat \oplus M \oplus (M^2) \oplus \hdots\).
\end{definition}

Although not crucial to understand our development, details about strong semi-rings and the mwp semi-ring can be found in \autoref{sec:app:mwp}, and the general construction of a semi-ring whose elements are matrices with coefficients in a different semi-ring -- so, in particular, \(\mathbb{M}(\textsc{mwp})\) -- is given in \autoref{sec:app:matrix}.

Below, we let \(V_1\), \(V_2\) be column vectors with values in \(\textsc{mwp}\), \(\alpha V_1\) to be the usual scalar product, and \(V_1 \oplus V_2\) to be defined component-wise.
We write \(\{_{i}^{\alpha}\}\) for the vector with \(0\) everywhere except for \(\alpha\) in its \(i\)th row, and \(\{_{i}^{\alpha} , _{j}^{\beta}\}\) for \(\{_{i}^{\alpha}\} \oplus \{_{j}^{\beta}\}\).
Given a matrix \(M\) and a vector \(V\), \(M \xleftarrow{j} V\) is \(M\) with the \(j\)th column vector replaced by \(V\).
We write \(\{_{i}^{\alpha} \rightarrow j\}\) for the matrix \(M\) with \(M_{ij} = \alpha\) and \(0\) everywhere else, and \(\var(\text{\pr|e|})\) for the set of variables in the expression \pr|e|.
In the analysis, the assumption is made that exactly \(n\) different variables are manipulated throughout the analyzed program, \(n\)-vectors are assigned to expressions and \(n \times n\) matrices are assigned to commands using rules reminded in \autoref{fig:orig-rules}~\cite[Section 5]{Jones2009}.

\begin{figure}
	\begin{subfigure}{\textwidth}
		\begin{centering}
			\begin{tabular}{c c c}
				\begin{prooftree}[small]
					\infer0[E1]{ \vdashJK \text{\pr|Xi|} : \{_{\text{\pr|i|}}^{m}\}}
				\end{prooftree}
				 & \hspace{1em} &
				\begin{prooftree}[small]
					\infer0[E2]{ \vdashJK \text{\pr{e}} : \{ _{\text{\pr|i|}}^{w} \mid \text{\pr|Xi|} \in \var(\text{\pr{e}}) \}}
				\end{prooftree}
				\\[2em]
				\begin{prooftree}[small]
					\infer0{\vdashJK \text{\pr{e1}} : V_1}
					\infer0{\vdashJK \text{\pr{e2}} : V_2}
					\infer[left label={$\star\in\{+, -\}$}]2[E3]{\vdashJK \text{\pr|e1 $\star$ e2|} : pV_1 \oplus V_2}
				\end{prooftree}
				 &              &
				\begin{prooftree}[small]
					\infer0{\vdashJK \text{\pr{e1}} : V_1}
					\infer0{\vdashJK \text{\pr{e2}} : V_2}
					\infer[left label={$\star\in\{+, -\}$}]2[E4]{\vdashJK \text{\pr|e1 $\star$ e2|} : V_1 \oplus pV_2}
				\end{prooftree}
			\end{tabular}

			\caption{Rules for assigning vectors to expressions}
			\label{fig:rules-expressions}
		\end{centering}
	\end{subfigure}
	\\[2em]
	\begin{subfigure}{\textwidth}
		\begin{centering}
			\begin{prooftree}[small]
				\hypo{ \vdashJK \text{\pr{e}} : V}
				\infer1[A]{\vdashJK \text{\pr|Xj = e|} : \umat \xleftarrow{\text{\pr|j|}} V}
			\end{prooftree}
			\hspace{2em}
			\begin{prooftree}[small]
				\hypo{ \vdashJK \text{\pr{C1}} : A}
				\hypo{ \vdashJK \text{\pr{C2}} : B}
				\infer2[C]{\vdashJK \text{\pr{C1 ; C2}} : A \otimes B}
			\end{prooftree}
			\\[2em]
			\begin{prooftree}[small]
				\hypo{ \vdashJK \text{\pr{C1}} : A}
				\hypo{ \vdashJK \text{\pr{C2}} : B}
				\infer2[I]{\vdashJK \text{\pr|if b then C1 $\text{\hspace{.2em}}$ else C2|} : A \oplus B} %
			\end{prooftree}
			\\[2em]
			\begin{prooftree}[small]
				\hypo{\vdashJK \text{\pr|C|} : M}
				\infer[left label={\(\forall i, M_{ii}^* = m\)}]1[L]{\vdashJK \text{\pr|loop Xl \{C\}|} : M^* \oplus \{_{\text{\pr|l|}}^{p}\rightarrow j \mid \exists i, M_{ij}^* = p\}}
			\end{prooftree}
			\\[2em]
			\begin{prooftree}[small]
				\hypo{\vdashJK \text{\pr|C|} : M}
				\infer[left label={\(\forall i, M_{ii}^* = m\) and \(\forall i, j, M^*_{ij} \neq p\)}]1[W]{\vdashJK \text{\pr|while b do \{C\}|} : M^*}
			\end{prooftree}
			\caption{Rules for assigning matrices to commands}
			\label{fig:rules-commands}
		\end{centering}
	\end{subfigure}
	\caption{Original (\enquote{\textsc{J}ones-\textsc{K}ristiansen}) rules}
	\label{fig:orig-rules}
\end{figure}

The intuition is that if \(\vdashJK \text{\pr|C|}:M\) can be derived using these rules, then all the values computed by \pr|C| will grow at most polynomially \wrt its inputs: this is the core and powerful result of the original paper~\cite[Theorem 5.3]{Jones2009}.
Furthermore, the coefficient at \(M_{\text{\pr|ij|}}\) carries quantitative information about the way \pr|Xi| depends on \pr|Xj|, knowing that \(0\)- and \(m\)-flows are harmless and without constraints, but that \(w\)- and \(p\)- flows are more harmful \wrt polynomial bounds and need to be handled with care, particularly when used in loops -- hence the condition on the L and W rules.
Although simple in appearance, the proof techniques are far from trivial, but the relative simplicity of the derivation and of the matrices manipulated make the analysis flexible and easy to carry.

In fact, the assignment may fail in case of loops---that is, some programs may not be assigned a matrix---, as at least one of the variables used in the body of the loop may depend \enquote{too strongly} upon another, making it impossible to ensure polynomial bounds on the loop itself (as iterating the body can lead to super-polynomial dependencies).

To capture a larger class of programs, the calculus used to assign a matrix to a program -- that corresponds to a proof search in a derivation system -- is non-deterministic.
As a consequence, multiple matrices---hence, multiple polynomial bounds---may be assigned to the same program.

We will use the following example (of \enquote{iteration-dependent} loop~\cite[Example 3.4]{Jones2009}) as a common basis to discuss our improvements.
\begin{example}
	\label{ex:start}
	Consider the command \pr|loop X3{X2 = X1 + X2}|.
	The body of the loop -- the expression \pr|X1 + X2| -- admits 3 different derivations that we name \(\pi_0\), \(\pi_1\) and \(\pi_2\):

	{\center

	\begin{prooftree}[small]
		\infer0[E1]{ \vdashJK \text{\pr|X1|} : \mat{m\\0\\0}}
		\infer0[E1]{ \vdashJK \text{\pr|X2|} : \mat{0\\m\\0}}
		\infer2[E3]{ \vdashJK \text{\pr|X1 + X2|} : \mat{p\\m\\0}}
	\end{prooftree}
	\hfill
	\begin{prooftree}[small]
		\infer0[E1]{ \vdashJK \text{\pr|X1|} : \mat{m\\0\\0}}
		\infer0[E1]{ \vdashJK \text{\pr|X2|} : \mat{0\\m\\0}}
		\infer2[E4]{ \vdashJK \text{\pr|X1 + X2|} : \mat{m\\p\\0}}
	\end{prooftree}
	\\[2em]
	\begin{prooftree}[small]
		\infer0[E2]{ \vdashJK \text{\pr|X1 + X2|} : \mat{w\\w\\0}}
	\end{prooftree}

	}

	From \(\pi_0\), the derivation of \pr|loop X3{X2 = X1 + X2}| can be completed, but since the L rule requires to have only \(m\) coefficients on the diagonal, \(\pi_1\) cannot be used to complete the derivation, because of the \(p\) coefficient in a box below:

	\begin{prooftree}[small]
		\hypo{}
		\ellipsis{\(\pi_0\)}{ \vdashJK \text{\pr|X1 + X2|} : \mat{p\\m\\0}}
		\infer1[A]{ \vdashJK \text{\pr|X2 = X1 + X2|} : \mat{m&p&0 \\ 0&m&0 \\0&0&m}}
		\infer1[L]{ \vdashJK \text{\pr|loop X3 \{X2 = X1 + X2\}|} : \mat{m&p&0 \\ 0&m&0 \\0&p&m}}
	\end{prooftree}
	\hfill
	\begin{prooftree}[small]
		\hypo{}
		\ellipsis{\(\pi_1\)}{ \vdashJK \text{\pr|X1 + X2|} : \mat{m\\p\\0}}
		\infer1[A]{ \vdashJK \text{\pr|X2 = X1 + X1|}: \mat{m&m&0 \\ 0& %
				\boxed{p} &0 \\0&0&m}}
	\end{prooftree}

	Similarly, because of the \(w\) coefficient on the diagonal after applying A, \(\pi_2\) cannot be used to complete the derivation either, and hence only one derivation for this command holds.
	Note that in general, multiple derivations can exist and that this \enquote{indeterminacy}~\cite[Section 8]{Jones2009} is needed to capture as many programs as possible.
\end{example}

\section{"Taming" non-determinism and non-termination}

The first two improvements over the existing analysis we offer are to:
\begin{enumerate}[nolistsep,noitemsep]
	\item \enquote{internalize} the non-determinism, so that at most one matrix per command is produced,
	\item \enquote{internalize} the failure, so that at least one matrix per command is produced.
\end{enumerate}
These changes were introduced first to obtain an efficient (or, actually, to simply enable an) implementation, but they came with by-products. %
Indeed, the naive approach consisting in producing a list of all possible matrices corresponding to all the non-deterministic choices (and removing those matrices for which the analysis fails) would result in a very slow implementation even for small programs.

To represent non-determinism, we use in the matrices \emph{functions from choices to coefficients in \(\textsc{mwp}\)} instead of simply coefficients in \(\textsc{mwp}\).
This is explained by the following remark -- made formal in \autoref{sec:app:choice}: the overall analysis produces a function from a space of choices \(C\) to the space \(\mathbb{M}(\textsc{mwp})\) of matrices over the mwp semi-ring, \ie it results in a function \(C\rightarrow \mathbb{M}(\textsc{mwp})\). But there is a semi-ring isomorphism between \(C\rightarrow \mathbb{M}(\textsc{mwp})\) and \(\mathbb{M}(C\rightarrow \textsc{mwp})\), \ie matrices whose coefficients are functions from choices to the mwp semi-ring. We use this, together with a clever representation of the space \(C\rightarrow \mathbb{M}(\textsc{mwp})\) to provide an alternative formalism allowing for more efficient implementation. Moreover, compacting all the possible derivations into one matrix results in a gain of space and time as different matrices obtained from different choices are \emph{more or less} the same, \ie they usually differ only on a few coefficients, leading to a quite compact representation. As a side-product, this also allows the user to be presented with different polynomial bounds, so that they can pick the one that suits their needs. %

Concerning failure, we extend the mwp semi-ring with a special value \(\infty\); one key point is that the resulting structure is \emph{not} a \emph{strong} semi-ring--as opposed to mwp or \(\mathbb{M}(\textsc{mwp})\)--because the latter structure requires the equality \(0\times \infty=0\) to hold while we need \(0\times \infty=\infty\) to avoid overlooking some super-polynomial computations: if part of the program computes an exponential value but then throws it away, \(0\times \infty=0\) would hide the super-polynomial computation, resulting in an incorrect analysis\footnote{Here we can be a bit more detailed: while throwing away the infinite coefficient would hide the super-polynomial computation, it would not contradict the \emph{ultimately} polynomial dependency of the values w.r.t. the inputs. As such, \(0\times \infty=0\) could still be used to bounds values, at the cost of losing the bounds on time and space usage for terminating programs. A modular implementation allowing to decide which structure to use in under progress.}. This way of representing failure also has the advantage of being local, so that which input variable impacts which variable in a non-polynomial way can be precisely pinpointed.
We believe this feature can be of crucial use in a situation where some variables are known to be of small size, hence where a non-polynomial bound \emph{on particular input variables} is acceptable.

Taken together, our improvements insure that exactly one matrix will always be assigned to a program, but also gives an opportunity to chose between \enquote{the lesser of two evils} when it fails: if two derivations produce \(\infty\) coefficients \emph{on different flows}, the user could decide to privilege one over the other based on knowledge about the inputs' sizes.

We give in \autoref{fig:new-rules} the alternative system we are introducing in full, but will gently discuss it though the remaining parts of this section and in \autoref{composition}: note that the A, C and I rules are unchanged (even if the sum and product are in a different semi-ring) and that the call rule is new.

\begin{figure}
	\begin{subfigure}{\textwidth}
		\begin{centering}
			\begin{prooftree}[small]
				\infer[left label={$\star\in\{+, -\}$}]0[E\(^{\textsc{A}}\)]{\vdash \text{\pr|Xi $\star$ Xj|} : (0 \mapsto \{_{\text{\pr|i|}}^{m}, _{\text{\pr|j|}}^{p}\}) \oplus (1 \mapsto \{_{\text{\pr|i|}}^{p}, _{\text{\pr|j|}}^{m}\}) \oplus (2 \mapsto \{_{\text{\pr|i|}}^{w}, _{\text{\pr|j|}}^{w}\})}
			\end{prooftree}
			\\[2em]
			\begin{prooftree}[small]
				\infer0[E\(^{\textsc{M}}\)]{ \vdash \text{\pr|Xi * Xj|} : \{ _{\text{\pr|i|}}^{w}, _{\text{\pr|j|}}^{w} \}}
			\end{prooftree}
			\caption{New rules for assigning vectors to expressions}
			\label{fig:new-rules-expressions}
		\end{centering}
	\end{subfigure}

	\begin{subfigure}{\textwidth}
		\begin{centering}
			\begin{prooftree}[small]
				\hypo{ \vdash \text{\pr{e}} : V}
				\infer1[A]{\vdash \text{\pr|Xj = e|} : \umat \xleftarrow{\text{\pr|j|}} V}
			\end{prooftree}
			\hfill
			\begin{prooftree}[small]
				\hypo{ \vdash \text{\pr{C1}} : A}
				\hypo{ \vdash \text{\pr{C2}} : B}
				\infer2[C]{\vdash \text{\pr{C1 ; C2}} : A \otimes B}
			\end{prooftree}
			\hfill
			\begin{prooftree}[small]
				\hypo{ \vdash \text{\pr{C1}} : A}
				\hypo{ \vdash \text{\pr{C2}} : B}
				\infer2[I]{\vdash \text{\pr|if b then C1 $\text{\hspace{.2em}}$ else C2|} : A \oplus B} %
			\end{prooftree}
			\\[2em]
			\begin{prooftree}[small]
				\hypo{ \vdash \text{\pr|C|}: M }
				\infer1[L\(^{\infty}\)]{\vdash \text{\pr|loop Xl \{C\}|} : M^* \oplus \{_{j}^{\infty} \rightarrow j \mid M^*_{jj} \neq m\} \oplus \{_{\text{\pr|l|}}^{p}\rightarrow j \mid \exists i, M_{ij}^* = p\} } %
			\end{prooftree}
			\\[2em]
			\begin{prooftree}[small]
				\hypo{ \vdash \text{\pr|C|}: M }
				\infer1[W\(^{\infty}\)]{\vdash \text{\pr|while b do \{C\}|} : M^* \oplus \{_{j}^{\infty} \rightarrow j \mid M^*_{jj} \neq m\} \oplus \{_{i}^{\infty}\rightarrow j \mid M_{ij}^* = p\} } %
			\end{prooftree}
			\\[2em]
			\begin{prooftree}[small]
				\Infer0[\(\mathrm{call}\)]{\vdash \text{\pr|Xi = F(X1,$\hdots$, XN)|} : 1 \xleftarrow{\text{\pr|i|}}((0 \mapsto M(f)_0) \oplus \dots \oplus (k \mapsto M(f)_k))}
			\end{prooftree}
		\end{centering}
		\caption{New rules for assigning matrices to commands}
		\label{fig:new-rules-commands}
	\end{subfigure}
	\caption{New rules}
	\label{fig:new-rules}
\end{figure}

\subsection{Choice data flow semi-rings}

The first step towards our \enquote{internalization of choice} is to design the correct semi-ring.
We start by reasoning abstractly, the detail of this construction is given in \autoref{sec:app:choice}.
Given a strong semi-ring \(\mathbb{S}\), we define \(\mathbb{M} (\mathbb{S})\) to be the strong semi-ring whose elements are matrices with coefficients in \(\mathbb{S}\) (\autoref{lem:matrices}), similarly to the matrix algebra of \autoref{def:mwp-matrix-alg}.
We also define \(A \to \mathbb{S}\) to be the strong semi-ring whose elements are functions from a set (of choices) \(A\) to \(\mathbb{S}\) (\autoref{lem:functions}).
We furthermore observe (\autoref{lem:semi-ring-iso}) %
that for all set \(A\) and strong semi-ring \(\mathbb{S} \), \(\mathbb{M} (A \to \mathbb{S} )\) and \(A \to \mathbb{M} (\mathbb{S} )\) are isomorphic (\autoref{def:iso}). By chosing \(A=\prod_{i=1}^p A_i\), it follows that there exists an isomorphism
\[
	\mathbb{M} (\prod_{i=1}^p A_i \to \mathbb{S} ) \cong \prod_{i=1}^p A_i \to \mathbb{M} (\mathbb{S} )
\]
for all family of sets \((A_i)_{i=1,\dots,p}\), using the usual cartesian product of sets.
This dual nature of the semi-ring considered will be useful:
\begin{itemize}[nolistsep,noitemsep]
	\item we implement the analysis by assigning elements of \(\mathbb{M} (\prod_{i=1}^p A_i \to \textsc{mwp} )\), this allows for a more efficient implementation by using some clever representation of elements of \(\prod_{i=1}^p A_i \to \textsc{mwp}\) detailed in \autoref{sec:implementation};
	\item we use the representation of the resulting matrix \(M\) as an element of \(\prod_{i=1}^p A_i \to \mathbb{M} (\textsc{mwp} )\) to produce, from an \emph{assignment} \(\alpha=(a_1,a_2,\dots,a_p)\in \prod_{i=1}^p A_i\), a matrix \(M[\alpha]\in \mathbb{M} (\textsc{mwp} )\), recovering the \emph{mwp}-flow that would have been computed by making the choices \(a_1,a_2,\dots\) in the derivation.
\end{itemize}

\begin{remark}
	As the unique degree of non-determinism in the rules to assign a matrix to commands is 3 at this point (\cf \autoref{ex:start}), our modification of the analysis flow consists simply (for the moment) in recording the different choices by letting \(A_i = \{0, 1, 2\}\) for all \(i=1,\dots,p\) where \(p\) is the number of times a choice had to be taken. Note that in a later section, other sets \(A_i\) will be used in order to deal with \emph{function calls}.
\end{remark}

\begin{example}
	Re-using the derivations \(\pi_0\), \(\pi_1\) and \(\pi_2\) from \autoref{ex:start}, we can now represent the three vectors \(\mat{p\\m\\0}\), \(\mat{m\\p\\0}\) and \(\mat{w\\w\\0}\) with a single vector
	\[
		\mat{\{ 0 \mapsto p, 1 \mapsto m, 2 \mapsto w\} \\ \{0 \mapsto m, 1 \mapsto p, 2 \mapsto w\} \\ 0}
	\]

	Where we make the abuse of notation of writing \(0\) for \(\{0 \mapsto 0, 1 \mapsto 0, 2 \mapsto 0\}\).\footnote{The implementation supports both coefficients from \(\textsc{mwp}\) \emph{and} coefficients from \(\{0, 1, 2\}^m \to \textsc{mwp}\), \cf \eg \href{https://seiller.github.io/pymwp/demo/\#basics_assign_expression.c}{a simple assignment \texttt{assign\_expression} example}.}
	Since, in particular\footnote{This is a variant of \autoref{lem:semi-ring-iso}. While the latter lemma is stated for an algebra of square matrices, a similar result holds for rectangular matrices of a fixed size; the algebraic structure is no longer that of a semi-ring as rectangular matrices do not possess a proper multiplication, but the proof can be adapted to show the existence of an isomorphism of modules between the considered spaces.}, \(\mathbb{M}(\{0, 1, 2\} \to\textsc{mwp}) \cong \{0, 1, 2\} \to \mathbb{M}(\textsc{mwp})\), the obtained vector can be rewritten as \(0 \mapsto \mat{p\\m\\0}, 1 \mapsto \mat{m\\p\\0}, 2 \mapsto \mat{w\\w\\0}\).

\end{example}

Our derivation system replaces the E3 and E4 rules with a single rule E\(^{\textsc{A}}\) (for \enquote{additive}), and imposes an additional restriction on E2, thus giving E\(^{\textsc{M}}\) (for \enquote{multiplicative}), so that it is used \emph{only} when E1 followed by E2 or E3 cannot be applied.

The implementation of binary additive operators (\(-\) and \(+\)) with E\(^{\textsc{A}}\) captures all possible choices for distinct operands and merges \pr|i| and \pr|j| into a single coefficient when \(\text{\pr|i|}= \text{\pr|j|}\).
Binary multiplication is handled by applying the E\(^{\textsc{M}}\) rule -- note that the application of E2 to additive operators in the original system is still handled by the last choice present in E\(^{\textsc{A}}\).
Given this need to treat binary operations differently, based on operators and combinations of operands, more work is needed to handle statements of greater arity.
As the implementation already processes abstract syntax trees of \texttt{C} commands recursively, handling operations of greater arity will require implementing additional recursive steps, but we do not expect that to be problematic conceptually or at the level of implementation.
At the light of this reflection, and knowing that there is no benefit in applying E2 to a single variable, as it result in a \(w\) coefficient being applied in lieu of a lesser \(m\) coefficient, it is easy to observe that E\(^{\textsc{A}}\) and E\(^{\textsc{M}}\) are as expressive as E1, E2, E3 and E4 taken together -- something we will be using when proving the equi-expressiveness of our system (\autoref{lem:equi-expr}).

\subsection{Representing failure with an \enquote{infinity} coefficient}

The original analysis would stop whenever a non-polynomial flow was detected, putting an end to the chosen strategy (\ie set of choices) and restarting from scratch with another one.
We will now discuss the fact that every derivation can be completed even in the presence of non-polynomial flows, which constitutes our second improvement.
This is done by first extending the mwp semi-ring with a new element. While this approach results in derivations for program where some variables \emph{are not} polynomially related to their inputs, we argue that pinpointing which variables are \enquote{faulty} from within the analysis can have benefits.

The first step is to incorporate a top element
\(\infty\) into our semi-rings to represent undefined elements. %
The semi-ring \(\textsc{mwp}^\infty\) we will be using is hence \((\textsc{mwp} \cup \{\infty\}, 0, m, +^{\infty}, \times^{\infty})\), with \(\infty > \alpha\) for all \(\alpha \in \textsc{mwp}\), \(+^{\infty}=\max\) as before, and \(\alpha \times^{\infty} \beta = 0\) if \(\alpha, \beta \neq \infty\) and \(\alpha\) or \(\beta\) is \(0\), \(\max(\alpha, \beta)\) otherwise.
This different condition in the definition of \(\times^{\infty}\) insures that once non-polynomial flows have been detected, they cannot be erased (as \(\infty \times^{\infty} 0 = \infty\)), but comes at the price of the strength of the semi-ring (the details are discussed in \autoref{sec:app:partiality}).

Below, we will work with \(\mathbb{M}(\textsc{mwp}^\infty)\), write \(\times\) for \(\times^{\infty}\) and similarly for \(+\), and remind the reader that we write \(\{_{i}^{\alpha} \rightarrow j\}\) for the matrix \(M\) with \(M_{ij} = \alpha\) and \(0\) everywhere else.
The only cases where the original analysis may fail is if the side condition of L or W (\autoref{fig:orig-rules}) are not met; we now replace those by the rules L$^\infty$ and W$^\infty$ of \autoref{fig:new-rules}, with no side condition.

Those rules, which can always be applied, simply replace the problematic coefficients with \(\infty\). Note that in the cases for which the original rule is applicable, the results coincide.
This will be essential to prove that our modified analysis is coherent with Jones and Kristiansen's original approach (\autoref{lem:equi-expr}).

\subsection{Merging the two improvements: illustration with operations}

We introduced and discussed the deviations from the original system for the \enquote{axiomatic} / \enquote{expression} (E\(^{\textsc{A}}\), E\(^{\textsc{M}}\)) and \enquote{loop} rules (L$^\infty$ and W$^\infty$), but remains to briefly discuss the rules for assignment (A), \pr|if| (I) and the composition (C), that remained unchanged.
Those rules are the place where both improvements meet.
Mathematically speaking, adopting the semi-ring defined over matrices using coefficients in \(\{0, 1, 2\}^m \to \textsc{mwp} \cup \{\infty\}\) is fairly simple, but computationally speaking, simple operations like multiplication and addition of matrices become very costly and memory-demanding.
This became particularly problematic when keeping a usable implementation in mind, and is illustrated below.

\begin{example}
	In our new system, consider the following derivation:
	\begin{center}

		\begin{prooftree}[small]
			\infer0[E\(^{\textsc{A}}\)]{\vdash \text{\pr|X1 + X2|} : V}
			\infer1[A]{\vdash \text{\pr|X1 = X1 + X2|} : 1 \xleftarrow{1} V}
			\infer0[E\(^{\textsc{A}}\)]{\vdash \text{\pr|X1 - X3|} : V'}
			\infer1[A]{\vdash \text{\pr|X1 = X1 - X3|} : 1 \xleftarrow{1} V'}
			\infer2[I]{\vdash \text{\pr|if b then \{X1 = X1 + X2\} else \{X1 = X1 - X3\}|} : (1 \xleftarrow{1} V) + (1 \xleftarrow{1} V')}
		\end{prooftree}

	\end{center}

	with

	\begin{align*}
		V                  & = 0 \mapsto \{_{\text{\pr|1|}}^{m}, _{\text{\pr|2|}}^{p}\} \oplus 1 \mapsto \{_{\text{\pr|1|}}^{p}, _{\text{\pr|2|}}^{m}\} \oplus 2 \mapsto \{_{\text{\pr|1|}}^{w}, _{\text{\pr|2|}}^{w}\}         \\
		V'                 & = 0 \mapsto \{_{\text{\pr|1|}}^{m}, _{\text{\pr|3|}}^{p}\} \oplus 1 \mapsto \{_{\text{\pr|1|}}^{p}, _{\text{\pr|3|}}^{m}\} \oplus 2 \mapsto \{_{\text{\pr|1|}}^{w}, _{\text{\pr|3|}}^{w}\}         \\
		1 \xleftarrow{1} V & = \mat{m                                                                                                                                                                                   & 0 & 0 \\ 0 & m & 0 \\ 0 & 0 & m} \xleftarrow{1} V \cong \mat{(0 \mapsto m) \oplus (1 \mapsto p) \oplus (2 \mapsto w) & 0 & 0 \\ (0 \mapsto p) + (1 \mapsto m) + (2 \mapsto w) & m & 0 \\ 0 & 0 & m} \\
		1 \xleftarrow{1} V & = \mat{m                                                                                                                                                                                   & 0 & 0 \\ 0 & m & 0 \\ 0 & 0 & m } \xleftarrow{1} V' \cong \mat{(0 \mapsto m) \oplus (1 \mapsto p) \oplus (2 \mapsto w) & 0 & 0 \\ 0 & m & 0 \\ (0 \mapsto p) \oplus (1 \mapsto m) \oplus (2 \mapsto w) & 0 & m}
	\end{align*}
	Now, to perform the addition required by the I rule, some care is needed: indeed, the choices in the left branch of the derivation are independent from the choices in the right branch, and we must use coefficients in \(\{0, 1, 2\}^2 \to \textsc{mwp}\) to represent the \(2^3\) choices.
	Assuming the choice in the left branch is first, we obtain \eg for the beginning of the top-left coefficient (the complete coefficient will be given below, once we introduced a more compact notation):
	\[(0 \mapsto (0 \mapsto (m + m = m)))
		+ (0 \mapsto (1 \mapsto (m + p = p)))
		+ (0 \mapsto (2 \mapsto (m + w = w)))
	\]
	Writing \(ab \mapsto\) for \(a \mapsto b \mapsto\), with \(a, b \in \{0, 1, 2\}\), and \(a\square \mapsto\) (\resp \(\square a \mapsto\)) if the second (\resp first) choice has no impact on the resulting coefficient, we can let:
	\[ A = 00\mapsto m + 01 \mapsto p + 02 \mapsto w + 1 \mapsto p + 20 \mapsto w + 21 \mapsto p + 22 \mapsto w\]
	to obtain
	\[(1 \xleftarrow{1} V) + (1 \xleftarrow{1} V') = \mat{A & 0 & 0 \\ (0\square \mapsto p) + (1\square \mapsto m) + (2\square \mapsto w) & m & 0 \\ (\square 0 \mapsto p) \oplus (\square 1 \mapsto m) \oplus (\square 2 \mapsto w) & 0 & m}\]

	Although the presentation and numbering diverge a bit, the example at \url{https://seiller.github.io/pymwp/demo/#improvement_paper_example3.c} can help the curious reader to check that the implementation reflects this derivation correctly.

\end{example}

\begin{example}
	Re-using \autoref{ex:start}, we now obtain in our new system a derivation that assign to \pr|loop X3 {X2 = X1 + X2}| the unique matrix
	\[
		\mat{m & (0 \mapsto p) + (1 \mapsto m) + (2 \mapsto w) &0 \\ 0 & (0 \mapsto m) + (1 \mapsto \infty) + (2 \mapsto \infty) & 0 \\ 0 & (0 \mapsto p) + (1 \mapsto 0) + (2 \mapsto 0) &m}
	\]
	where we observe that
	\begin{enumerate*}\item only one choice (0) -- one assignment -- gives a matrix without \(\infty\) coefficient, corresponding to the fact that, in the original system, only \(\pi_0\) could be used to complete the proof, \item the choice impact the matrix only \emph{locally}, the coefficients being \emph{mostly} the same, independently from the choice, \item the influence of \pr|X2| on itself is where possible non-polynomial growth rates lies, as the \(\infty\) coefficient are in the second column, second row.
	\end{enumerate*}
	This example was not implemented, as \pr|loop| is not a standard \texttt{C} operator, but is currently being implemented as a restricted form of \prc{for} loop (\cf \url{https://github.com/seiller/pymwp/issues/5}).
\end{example}

We are now in possession of all the material and intuitions to state the correspondence between our approach and the one of Jones and Kristiansen.

\begin{lemma}
	\label{lem:equi-expr}
	Given a program \(P\), there is a single matrix \(M\in \mathbb{M}(\{0,1,2\}^p\rightarrow\textsc{mwp}^\infty)\) such that \(P\vdash M\), \ie the system is deterministic. Moreover, for any assignment \(\alpha=(a_1,\dots,a_p)\in A^p\), we have that
	\[ P\vdashJK M[\alpha] \text{ if and only if } M[\alpha] \in \mathbb{M}(\textsc{mwp}). \]
\end{lemma}

This shows that the performed analyses coincide, as \( M[\alpha] \in \mathbb{M}(\textsc{mwp})\) implies that no \(\infty\) coefficient occurs in it.
However, our alternative definition should be understood as an important improvement, as it allows for a more efficient implementation (\autoref{sec:implementation}).
But before discussing the efficiency of the implementation, we will now explain the natural but important extension to function calls enabled by our alternative formalism.

\section{Extending the analysis with function calls}
\label{composition}

We begin by extending the syntax presented in \autoref{subsec:language} by adding \emph{function declarations} \pr|F| \(\coloneqq\) \pr|f(X1, $\hdots$, XN){C; return R}| and a command that performs a function call and assign its return value to a variable \pr|Xi = F(X1, $\hdots$, XN)|\footnote{Function calls that discard the output could also be dealt with easily, but are vacuous in our effect-free language}.
A \emph{program} is now a series of function declarations, with one of them called \pr|main| with \(N = 0\), and such that all the commands of the form \pr|Xi = F(X1,|\(\hdots\)\pr|, XN)| refers to a function previously declared.
A \emph{chunk} is simply a series of commands inside a function declaration\footnote{Note that this implies that if a loop belongs to a chunk, then the entire loop body belongs to the chunk.}.

One of the key points of our contribution is the extension of the analysis to function calls, in a way that can be used in practice, as we handle a function \(f\) with a single analysis that stores a minimal amount of data for latter calls.
The principle is the following: given the matrix \(M(f)\) obtained from the analysis of the program computing \(f\), we store only the \(k\) choices for which no \(\infty\) coefficients appear, and then project them to only keep track of the different input/output behaviors, merging choices leading to the same result.
After this operation, we are left with a family \(M(f)_0, M(f)_1, \hdots,M(f)_k\) of matrices\footnote{To ease the presentation, the syntax considered here is restricted to functions with a single output value, so we actually have vectors in place of matrices here. But it is more natural to think in terms of matrices here, as the overall approach is valid in the more general setting in which functions may have several output values, and then the obtained objects are indeed matrices.} that should be understood as providing quantitative (\ie polynomial, weak polynomial, maximum, or zero) information about the dependency of output values w.r.t. input values. Now, the analysis of the command calling the function \(f\) is dealt with by the call rule of \autoref{fig:new-rules}.

Formally, we show that our definition of composition is coherent with the initial analysis as follows. We consider two \emph{chunks}: the first chuck \(P\) contains a call to a function \(f\), the second is obtained by replacing within \(P\) the call to the function \(f\) by inserting in its place the sequence of commands \(F\) computing \(f\). This second chunk is called \(P[F]\). We then prove that the matrix associated to \(P\) is \enquote{the same}\footnote{Here one has to consider equality up to some projections as the chuck \(F\) inserted in \(P\) may introduce new choices and use additional variables.}
Intuitively, this mechanism provides the expected result because the choices made in the chunk \(F\) do \emph{not} affect the context \(P[\cdot]\), and the variables used in the chunk \(F\) are \emph{not} used in the context \(P[\cdot]\) except for the return variables. %

More formally, let \(P\) be a chunk of program, containing a call to the function \(f\), and let \(F\) be the chunk computing the function \(f\). We define from \(P\) the \emph{context} \(P[\cdot]\), a chunk containing a hole \([\cdot]\) to be filled with the chunk \(F\), obtained as follows (supposing \(f\) has a single output variable).
\begin{itemize}[noitemsep,nolistsep]
	\item We remove the line with the function call, say \pr|Xi=f(X1, $\hdots$, XN);|.
	\item We add in place the following lines, where \pr|R|, \pr|Y1|, \(\hdots\), \pr|Yn| are fresh variables:
	      \begin{lstlisting}[language=C, belowskip=0.1 \baselineskip, aboveskip=0.1 \baselineskip]
	Y1 = X1;
	|$\dots$|
	YN = XN;
	|$[\cdot]$|
	Xi = R;
\end{lstlisting}
\end{itemize}

The code \(P[F]\) is then obtained by defining a chunk \(\tilde{F}\), and inserting it in place of the symbol \(\cdot\) in \(P[\cdot]\). The chunk \(\tilde{F}\) is obtained as follows from \(F\):
\begin{itemize}[nolistsep,noitemsep]
	\item the header is removed,
	\item the input variables of \(F\) are renamed to \pr|Y1|, \pr|Y2|, \(\hdots\), \pr|YN|,
	\item the variable returned by \(f\) is renamed to \pr|R|, the \prc|return| statement is removed,
	\item all other variables are renamed if needed to avoid using the same names as the variables in \(P[\cdot]\). We write the set of these variables \(V_{F}\).
\end{itemize}

\begin{example}
	Refer to \autoref{fig:example-inline} for a simple example of the code transformation for in-lining a function call.
	\begin{figure}
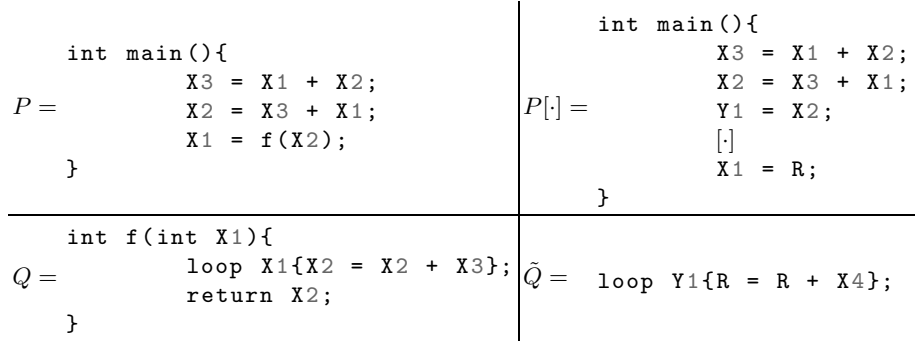

		{\centering

			\begin{tabular}{ l l | l l }
				\(P = \)        & \begin{lstlisting}[language=C]
int main(){
	X3 = X1 + X2;
	X2 = X3 + X1;
	X1 = f(X2);
}
\end{lstlisting}
				                &
				\(P[\cdot] = \) & \begin{lstlisting}[language=C]
int main(){
	X3 = X1 + X2;
	X2 = X3 + X1;
	Y1 = X2;
	|$[\cdot]$|
	X1 = R;
}
\end{lstlisting}
				\\ \hline
				\(Q = \)        & \begin{lstlisting}[language=C]
int f(int X1){
	loop X1{X2 = X2 + X3};
	return X2;
}
\end{lstlisting} &
				\(\tilde{Q} =\) & \begin{lstlisting}[language=C]
loop Y1{R = R + X4};
\end{lstlisting}
			\end{tabular}

		}
		\caption{A simple example of \enquote{inlining} a function call}
		\label{fig:example-inline}
	\end{figure}
\end{example}

Now, we can compute both matrices:
\begin{itemize}[nolistsep,noitemsep]
	\item \(M(P)\) where the line \pr|Xi=f(X1, $\hdots$, XN);| is analysed using the \(\mathrm{call}\) rule, and
	\item \(M(P[F])\).
\end{itemize}
We write \(\Pi_P(M(P[F]))\) the projection of \(M(P[F])\) onto the variables in \(P\) and \((1-\Pi_P)(M(P[F]))\) the projection of \(M(P[F])\) onto the variables \emph{not in} \(P\).

Some non-deterministic choices may appear within the (modified) chunk \(\tilde{F}\) inside \(P[F]\), \ie
\begin{itemize}[nolistsep,noitemsep]
	\item the coefficients of the matrix \(M(P)\) are elements of the semi-ring \(\prod_{i=1}^{p+1} A_i \rightarrow \mathbb{M}(\textsc{mwp})\), with one particular choice corresponding to the \(\mathrm{call}\) rule -- we write the corresponding index \(i_0\);
	\item the coefficients of \(P[F]\) are elements of the semi-ring \(\prod_{i=1}^{p+k} B_i\rightarrow \mathbb{M}(\textsc{mwp})\), where \(k\) choices are made within the chunk \(\tilde{F}\) -- we write the corresponding indexes \(j_1,j_2,\dots,j_k\) (note these are in fact consecutive indexes).
\end{itemize}
We note \(\pi: \{1,\dots,p+k\} \rightarrow \{1,\dots,p+1\}\) the projection of the choices in \(P[F]\) onto the corresponding choices in \(P\), \ie
\[
	\pi(j)=\left\{
	\begin{tabular}{ll}
		\(j\)     & \text{ if \(j<j_0\)}               \\
		\(i_0\)   & \text{ if \(j_0 \leqslant j<j_k\)} \\
		\(j-k+1\) & \text{ if \(j_k< j\)}
	\end{tabular}
	\right.
\]
We note that each matrix used as axiom in the function call corresponds to a specific assignment on indexes \(j_1,\dots,j_k\). We write \(\Psi: A_{i_0}\rightarrow \prod_{i=j_1}^{j_k} B_i\) the corresponding injection. This is extended to \(\bar{\Psi}: \prod_{i=1}^{p+1} A_{i}\rightarrow \prod_{i=0}^{p+k} B_i\) in a straightforward way.

We can now state the main theorem showing that the \(\mathrm{call}\) rule adequately analyses function calls.
\begin{theorem}
	For all assignment \(\alpha\) of \(\prod_{i=1}^{p+1} A_i\),
	\[ M(P)[\alpha]=(1-\Pi_P)(M(P[F]))[\bar{\Psi}(\alpha)] \]
	Moreover, for all assignment \(\beta\) of \(\prod_{i=0}^{p+k} B_i\) not in \(\mathrm{Im}(\bar{\Psi})\), the matrix \((1-\Pi_P)(M(P[F])[\beta])\) contains an infinite value.
\end{theorem}

\begin{proof}
	To prove this, we first notice that it is sufficient to prove it for the simplest chuck \(P\) containing only one command: \pr|Xi = f(X1, $\hdots$, XN)|. This is explained by the compositional nature of the analysis (a sequence of commands is simply assigned the product of the matrices of each individual command). Then, checking that the theorem holds in this case is a straightforward, though tedious (due to keeping track of all indices), computation.
\end{proof}

\section{Implementation of the analysis}
\label{sec:implementation}

The formulation of the extended mwp analysis exposed in the previous sections was also intended for implementation. As such, the choice of the representation of non-determinism -- for instance -- was also guided by our wish for a faster analysis, something not discussed in depth in our tool paper~\cite{Aubert2021f}, or \href{https://seiller.github.io/pymwp/}{our documentation}.
In this section, we expose some of the specific choices made in the implementation.

\subsection{Non-determinism, and the challenges to efficient calculations}

As explained in the above sections, the result of the analysis is a matrix with coefficients in a semi-ring of the form \(\prod_{i=1}^{p} A_i \rightarrow \mathbb{M}(\textsc{mwp})\)-- setting aside \(\infty\) coefficients for a moment.
To implement this correctly, we represent elements of this semi-ring as polynomials w.r.t. the generating set given by the functions \(\delta(i,j): \prod_{i=1}^{p} A_i \rightarrow {\textsc{mwp}}\) defined by \(\delta(i,j)(a_1,a_2,\dots,a_p)=m\) if \(a_j=i\) and \(\delta(i,j)(a_1,a_2,\dots,a_p)=0\) otherwise. \ie an element of \(\prod_{i=1}^{p} A_i \rightarrow {\textsc{mwp}}\) is represented as a polynomial \(\sum_{i=1}^{n} \alpha_i \prod_{j=1}^{k_i}\delta(a_{i,j},b_{i,j})\) with \(\alpha_i \in {\textsc{mwp}}\).

This basis have an important property: the monomials $\alpha_i \prod_{j=1}^{k_i}\delta(a_{i,j},b_{i,j})$ in a polynomial can be ordered in such a way that the product with another monomial is ordered. \ie if \(m\leqslant m'\) and both \(m\times n\) and \(m' \times n\) are non-zero, then \(m\times n\leqslant m'\times n\). This order is leveraged to obtain efficient algorithms for computing operations on the representation of coefficients, similar to what is done using Gröbner bases for computation of standard polynomials.
For instance, the algorithm for \href{https://github.com/seiller/pymwp/blob/746da71a5490c5f21ebc5643ea20822f78876959/pymwp/polynomial.py#L199}{multiplication of polynomials} makes use of the property above and proceeds as follows to compute the product of a polynomial \(P\) with \(\sum_{i=1}^{n} \alpha_i \prod_{j=1}^{k_i}\delta(a_{i,j},b_{i,j})\) (supposing the representation of \(P\) is ordered):
\begin{enumerate}[nolistsep,noitemsep]
	\item compute the products \(P_i= P\times \alpha_i \prod_{j=1}^{k_i}\delta(a_{i,j},b_{i,j})\) for all \(i\);
	\item compare and order a list \(L\) of all the first elements of those polynomials;
	      \item\label{algostep} append the smallest element to the result and remove it from the corresponding \(P_i\);
	\item insert the (new) first element of \(P_i\) to the list \(L\) if it exists;
	\item if \(L\) is non-empty, go back to step \ref{algostep}.
\end{enumerate}
This clever method has some very concrete consequences.
As an example, our \href{https://seiller.github.io/pymwp/demo/#other_explosion.c}{\texttt{explosion.c}} program calls the multiplication 11,907 times and could not be completed with a naive multiplication implementation.
More precise \href{https://seiller.github.io/pymwp/utilities/}{profiling} further exposes the need for this optimization.

\subsection{Infinite values cluttering the analysis, and difficulties to evaluate}

One very costly aspect of the analysis is the \emph{evaluation} step which takes a matrix with coefficients in \(\prod_{i=1}^{p} A_i \rightarrow {\textsc{mwp}}\) and checks all possible assignments \((a_1,\dots,a_p)\in \prod_{i=1}^{p} A_i\) to look for infinite coefficients. While this step is necessary (in one form or another) if one wishes to produce the actual mwp matrices certifying polynomial bounds (something needed at least once to allow for function calls), we implemented a specific data structure allowing to keep track of infinite assignments on the fly, thus allowing the analysis to provide a qualitative answer quickly. \Ie the analysis can ensure the existence or not of mwp-bounds \emph{without computing the corresponding matrix}.

This is implemented by a structure we called \href{https://github.com/seiller/pymwp/blob/946a5b44692325095392694950ed03807f059b52/pymwp/delta_graphs.py}{\texttt{delta\_graphs}}.
This is a graph whose vertices are monomials; the graph is populated during the analysis by adding those monomials that appear with an infinite coefficient -- \ie possible choices leading to \(\infty\) in the resulting matrix.
This graph is structured in layers: each layer corresponds to the size of the monomials it contains (the number of deltas \(\delta(i,j)\) is contains). The intuition is that a monomial -- or rather a list of deltas $\delta(\_,\_)$-- defines a subset of the space \(\prod_{i=1}^{p} A_i\); the less deltas in the monomial, the greater the subspace represented. (Note here that our intuitions come from the standard topological structure of spaces of infinite sequences, where such a monomial represents a \enquote{cylinder set}, \ie an element of the standard basis for open sets.) As we populate the delta\_graph, we create edges within a given layer to keep track of differences between monomials: we add an edge labeled \(i\) between two monomials if and only if they differ only on one delta \(\delta(\_,i)\) (\ie one is obtained from the other by replacing the first index of \(\delta(\_,i)\)). This is used to implement a \href{https://github.com/seiller/pymwp/blob/946a5b44692325095392694950ed03807f059b52/pymwp/delta_graphs.py#L274}{\enquote{fusion}} method on delta\_graphs which simplifies the structure: as soon as as a monomial \(m\) in layer \(n\) has \(\mathrm{Card}(A_i)-1\) outgoing edges labelled \(i\), we can remove all these monomials and insert a shorter monomial in layer \(n-1\) (obtained from \(m\) by simply removing \(\delta(\_,i)\)). This implements the fact that \(\sum_{k=0}^{\mathrm{Card}(A_i)-1} m \delta(k,j) = m\).

Remember the delta\_graph represents the subspace of assignments for which an infinite coefficient appeared. So if at some point the delta\_graph is completely simplified (\ie \enquote{fusions} to the graph with a unique monomial consisting in an empty list of \(\delta(\_,\_)\)), it means the whole space of assignments is represented and no mwp-bounds can be found. On the contrary, if the analysis ends with a delta\_graph different from the completely simplified one, it means at least one assignment exists for which no infinite coefficients appear, and therefore at least one mwp-bound exists.

\section{Future work}

We here provide some details on extensions of this work that we are currently working on, or that will be tackled in the near future.

The first natural line of work is the extension of the language analysed, in particular to accommodate other data structures.
While structures such as lists should not be problematic, dealing with pointer will certainly require more involved work, in particular to ensure the theoretical results obtained by Jones and Kristiansen hold, \ie that the obtained mwp-bounds are indeed correct. These extensions, together with the extension to function calls discussed in this paper, will then be added within our implementation of the analysis.

A second line of work that was already started is to implement the analysis in the Compcert compiler~\cite{Leroy2009}, which would allow for a formal certification of the polynomial bounds computed by the analysis using the Coq proof assistant~\cite{coqman}.
Some preliminary work in this direction was already done. In particular, it seems natural to use \href{https://compcertssa.gitlabpages.inria.fr}{\sc compcert-ssa}~\cite{Barthe2014} to be later used as stepping stone towards an implementation within {\texttt llvm} -- and if possible \href{https://www.cis.upenn.edu/~stevez/vellvm/}{\texttt certified-llvm}~\cite{Zhao2013} -- which would enable the analysis to programs written in other languages than {\texttt C}.

\bibliographystyle{splncs04}
\bibliography{standalone}

\appendix
\section{Technical Appendix on Semi-rings}
\label{app:sec:semi-ring}

\subsection{The mwp semi-ring}
\label{sec:app:mwp}

This subsection briefly recall semi-ring definition and proves that the mwp semi-ring is indeed a semi-ring.

\begin{definition}[Semi-ring]
	\label{def:semi-ring}
	A semi-ring \(\mathbb{S} =(S,0,1,+,\times )\) is specified by a set \(S\) and two binary operations \(+\) (addition) and \(\times \) (multiplication) such that \(\{0,1\} \in S\) and
	\begin{enumerate}
		\item \((S, 0,+)\) is a commutative monoid: the operation \(+\) is associative, commutative, and has \(0\) as the identity element,
		\item \((S, 1,\times )\) is a monoid: the operation \(\times \) is associative and has \(1\) as the identity element,
		\item the operation \(\times \) distributes with respect to \(+\): for all \(a, b, c \in S\), \(a \times (b + c) = a \times b + a \times c\) and \((b + c) \times a = b \times a + c \times a\)
	\end{enumerate}

	We call \(\mathbb{S} \) a \emph{strong} semi-ring if, additionally, \emph{\(0\) annihilates \(S\)}, i.e.\

	\begin{enumerate}[resume]
		\item \(0\times a=a\times 0=0\) for all \(a \in S\).
	\end{enumerate}
\end{definition}

\begin{lemma}[mwp-semi-ring]
	\label{lem:mwp-is-a-semiring}
	The tuple \((\{0,m,w,p\}, 0, m, +, \times )\), with
	\begin{itemize}
		\item \( 0 < m < w < p \),
		\item \(\alpha + \beta = \begin{dcases*}
			      \alpha & if \(\alpha \geqslant \beta \) \\
			      \beta  & otherwise
		      \end{dcases*}\)
		\item \(\alpha \times \beta =
		      \begin{dcases*}
			      \alpha + \beta & if \(\alpha \neq 0\) and \(\beta \neq 0\) \\
			      0              & otherwise
		      \end{dcases*}\)
	\end{itemize}
	is a \emph{strong} semi-ring.
\end{lemma}

\begin{proof}
	We prove that \((\{0,m,w,p\}, 0, m, +, \times )\) as defined respects the conditions of \autoref{def:semi-ring}.
	The proof is straightforward but detailed nevertheless.
	\begin{description}
		\item[\((\{0,m,w,p\}, 0, +)\) is a commutative monoid]
			We prove that \((\{0,m,w,p\}, +)\) is a commutative monoid by showing that it is associative, commutative, and has \(0\) as identity.

			\begin{description}
				\item[Associativity]
					\((\alpha + \beta ) + \gamma = \alpha + (\beta + \gamma)\)
					\begin{description}
						\item[Case 1: \(\alpha \geqslant \beta \geqslant \gamma\)]
							\begin{align*}
								         &  & \alpha                     & = \alpha                    \\
								\implies &  & \alpha + \gamma            & = \alpha + \beta            \\
								\implies &  & (\alpha + \beta ) + \gamma & = \alpha + (\beta + \gamma)
							\end{align*}

						\item[Case 2: \(\alpha \geqslant \gamma \geqslant \beta\)]
							\begin{align*}
								         &  & \alpha                     & = \alpha                    \\
								\implies &  & \alpha + \gamma            & = \alpha + \gamma           \\
								\implies &  & (\alpha + \beta ) + \gamma & = \alpha + (\beta + \gamma)
							\end{align*}

						\item[Case 3: \(\beta \geqslant \alpha \geqslant \gamma\)]
							\begin{align*}
								         &  & \beta                      & = \beta                     \\
								\implies &  & \beta + \gamma             & = \alpha + \beta            \\
								\implies &  & (\alpha + \beta ) + \gamma & = \alpha + (\beta + \gamma)
							\end{align*}

						\item[Case 4: \(\beta \geqslant \gamma \geqslant \alpha\)]
							\begin{align*}
								         &  & \beta                      & = \beta                     \\
								\implies &  & \beta + \gamma             & = \alpha + \beta            \\
								\implies &  & (\alpha + \beta ) + \gamma & = \alpha + (\beta + \gamma)
							\end{align*}

						\item[Case 5: \(\gamma \geqslant \alpha \geqslant \beta\)]
							\begin{align*}
								         &  & \gamma                     & = \gamma                    \\
								\implies &  & \alpha + \gamma            & = \alpha + \gamma           \\
								\implies &  & (\alpha + \beta ) + \gamma & = \alpha + (\beta + \gamma)
							\end{align*}

						\item[Case 6: \(\gamma \geqslant \beta \geqslant \alpha\)]
							\begin{align*}
								         &  & \gamma                     & = \gamma                    \\
								\implies &  & \beta + \gamma             & = \alpha + \gamma           \\
								\implies &  & (\alpha + \beta ) + \gamma & = \alpha + (\beta + \gamma)
							\end{align*}
					\end{description}

				\item[Commutative Property]
					Both cases are immediate:
					\begin{description}
						\item[Case 1: \(\alpha \geqslant \beta\)] \(\implies \alpha + \beta = \alpha = \beta + \alpha\)
						\item[Case 2: \(\beta \geqslant \alpha\)] \(\implies \alpha + \beta = \beta = \beta + \alpha\)
					\end{description}

				\item[Identity element is \(0\)]
					\[0 + 0 = 0 \qquad 0 + m = m \qquad 0 + w = w \qquad 0 + p = p\]
			\end{description}

		\item[\((\{0,m,w,p\}, m, \times )\) is a monoid]
			We now prove that \((\{0,m,w,p\}, m, \times )\) is a monoid by showing that it is associative, has \(m\) as identity, and has \(0\) as the annihilator.

			\begin{description}
				\item [Associativity]

				      \((\alpha \times \beta ) \times \gamma = \alpha \times (\beta \times \gamma )\)

				      \begin{description}
					      \item[Case 1: ] \(\alpha, \beta, \gamma \in \{m,w,p\} \)

						      \(\alpha \times \beta = \alpha + \beta \) Associativity of operation + is shown in the proof of the commutative monoid, \((\{0,m,w,p \}, +)\).

					      \item[Case 2: ] \(\alpha\), \(\beta\), or \(\gamma\) equals \(0\)

						      By definition of multiplication, the product is \(0\).
				      \end{description}

				\item[Identity element is \(m\)]
					\begin{alignat*}{4}
						 & 0 \times m &  & = 0 &  & = m \times 0 \\
						 & m \times m &  & = m &  & = m \times m \\
						 & w \times m &  & = w &  & = m \times w \\
						 & p \times m &  & = p &  & = m \times p
					\end{alignat*}

				\item[0 annihilates \(\{0,m,w,p\}\)]
					\begin{alignat*}{4}
						 & 0 \times 0 &  & = 0 &  & = 0 \times 0 \\
						 & m \times 0 &  & = 0 &  & = 0 \times m \\
						 & w \times 0 &  & = 0 &  & = 0 \times w \\
						 & p \times 0 &  & = 0 &  & = 0 \times p
					\end{alignat*}
			\end{description}

		\item [Distribution of multiplication over addition]
		      We conclude by proving that \(\times \) distributes over \(+\).

		\item[Right Distribution]
			\(\alpha \times (\beta + \gamma) = (\alpha \times \beta) + (\alpha \times \gamma)\)
			\begin{description}
				\item[Case 1: \(\beta \geqslant \gamma\)]
					\begin{align*}
						\implies &  & \alpha \times \beta            & = \alpha \times \beta                              \\
						\implies &  & \alpha \times (\beta + \gamma) & = (\alpha \times \beta ) + (\alpha \times \gamma )
					\end{align*}

				\item[Case 2: \(\gamma \geqslant \beta\)]
					\begin{align*}
						\implies &  & \alpha \times \gamma           & = \alpha \times \gamma                             \\
						\implies &  & \alpha \times (\beta + \gamma) & = (\alpha \times \beta ) + (\alpha \times \gamma )
					\end{align*}
			\end{description}

		\item[Left Distribution]
			\( (\alpha + \beta ) \times \gamma = (\alpha \times \gamma) + (\beta \times \gamma)\)
			\begin{description}
				\item[Case 1: \(\alpha \geqslant \beta\)]
					\begin{align*}
						\implies &  & \alpha \times \gamma            & = \alpha \times \gamma                             \\
						\implies &  & (\alpha + \beta ) \times \gamma & = (\alpha \times \gamma ) + (\beta \times \gamma )
					\end{align*}

				\item[Case 3: \(\beta \geqslant \alpha\)]
					\begin{align*}
						\implies &  & \beta \times \gamma             & = \beta \times \gamma                                       \\
						\implies &  & (\alpha + \beta ) \times \gamma & = (\alpha \times \gamma ) + (\beta \times \gamma ) \qedhere
					\end{align*}
			\end{description}

	\end{description}
\end{proof}

\subsection{Matrix Semi-ring}
\label{sec:app:matrix}

This subsection explains and details how matrices with coefficients in a semi-ring can be used to construct semi-rings.

\begin{lemma}
	\label{lem:matrices}
	Given a strong semi-ring \(\mathbb{S} = (S, 0, 1, +, \times )\), we define the tuple \(\mathbb{M} = (M,\mathbf{0} ,\mathsfbf{1} ,\oplus ,\otimes )\), with
	\begin{itemize}
		\item \(M\) the set of all \(n \times n\) matrices over \(S\), for all \(n \in \mathbb{N}\),
		\item \(\mathbf{0} \) defined by \(M = \mathbf{0} \) iff \(M_{ij} = 0\) for all \(i\) and \(j\),
		\item \(\mathsfbf{1} \) defined by \(M = \mathsfbf{1} \) iff \(M_{ij} = 1\) for \(i = j\), \(M_{ij} = 0\) otherwise,
		\item \(\oplus \) defined by \(C = A \oplus B\) iff \(C_{ij} = A_{ij} + B_{ij}\),
		\item \(\otimes \) defined by \(C = A \otimes B\) iff \(C_{ij} = \sum_{k=1}^{n} A_{ik} \times B_{kj}\),
	\end{itemize}
	is a strong semi-ring.
\end{lemma}

\begin{proof}
	We prove that \(\mathbb{M} = (M,\mathbf{0} ,\mathsfbf{1} ,\oplus ,\otimes )\) as defined respects the conditions of \autoref{def:semi-ring}.
	Let \(A\),\(B\),\(C\) be \(n \times n\) matrices over \(S\) where \(n \in \mathbb{N}\).
	\begin{description}
		\item[\( (M,\mathbf{0} ,\mathsfbf{1} ,\oplus )\) is a commutative monoid]
			We first prove that \((M, \oplus )\) is a commutative monoid by showing that it is associative, commutative, and has \(\mathbf{0} \) as identity.

			\begin{description}
				\item[Associativity]
					\((A \oplus B) \oplus C = A \oplus (B \oplus C)\) iff \( ((A \oplus B) \oplus C)_{ij} = (A \oplus (B \oplus C))_{ij}\) for all \(i\), \(j\).
					\begin{align*}
						((A \oplus B) \oplus C)_{ij} & = (A \oplus B)_{ij} + C_{ij}                                 \\
						                             & = (A_{ij} + B_{ij}) + C_{ij}                                 \\
						                             & = A_{ij} + (B_{ij} + C_{ij}) \tag{by associativity of \(+\)} \\
						                             & = A_{ij} + (B \oplus C)_{ij}                                 \\
						                             & = (A \oplus ( B \oplus C))_{ij}
					\end{align*}

				\item[Commutative Property]
					\(A \oplus B = B \oplus A\) iff \((A \oplus B)_{ij} = (B \oplus A)_{ij}\) for all \(i\), \(j\).
					\begin{align*}
						(A \oplus B)_{ij} & = A_{ij} + B_{ij}                                 \\
						                  & = B_{ij} + A_{ij} \tag{by commutativity of \(+\)} \\
						                  & = (B \oplus A)_{ij}
					\end{align*}

				\item[Identity element is \(\mathbf{0} \)]
					Let \(A = \mathbf{0} \), then \(A_{ij} = 0\) for all \(i\), \(j\), and \(\mathbf{0} \) is the identity element iff \(A_{ij} + B_{ij} = B_{ij}\) for all \(i\), \(j\)
					\begin{align*}
						(A \oplus B)_{ij} & = A_{ij} + B_{ij}                       \\
						                  & = 0 + B_{ij} \tag{by identity of \(+\)} \\
						                  & = B_{ij}
					\end{align*}
			\end{description}

		\item[\((M, \mathsfbf{1} , \otimes )\) is a monoid]
			We now prove that \((M, \otimes )\) is a monoid by showing that it is associative and has \(\mathsfbf{1} \) as identity.

			\begin{description}
				\item [Associativity]
				      \((A \otimes B) \otimes C = A \otimes (B \otimes C)\) iff \(((A \otimes B) \otimes C) _{ij} = (A \otimes (B \otimes C))_{ij}\) for all \(i\), \(j\).
				      \begin{align*}
					      ((A \otimes B) \otimes C) _{ij} & = (\sum_{k=1}^{n} A_{ik} \times B_{kj}) \otimes C                                                  \\
					                                      & = \sum_{l=1}^{n} (\sum_{k=1}^{n} A_{ik} \times B_{kj})_{il} \times C_{lj}                          \\
					                                      & = \sum_{l=1}^{n} \sum_{k=1}^{n} (A_{ik} \times B_{kl}) \times C_{lj}                               \\
					                                      & = \sum_{k=1}^{n} \sum_{l=1}^{n} A_{ik} \times (B_{kl} \times C_{lj})\tag{by assoc. of \(\times \)} \\
					                                      & = \sum_{k=1}^{n} A_{ik} \times (\sum_{l=1}^{n} B_{il} \times C_{lj})_{kj}                          \\
					                                      & = A \otimes (\sum_{l=1}^{n} B_{il} \times C_{lj})                                                  \\
					                                      & = (A \otimes ( B \otimes C))_{ij}
				      \end{align*}

				\item[Identity element is \(\mathsfbf{1} \)] \(A \otimes B = B \) and \(B \otimes A = B \) where \(A = \mathsfbf{1} \) iff \(A_{ij} = 1\) for \(i = j\) and \(A_{ij} = 0\) otherwise.
					\begin{align*}
						(A \otimes B)_{ij} & = \sum_{k=1}^{n} A_{ik} \times B_{kj}                                                                      \\
						                   & = (A_{ii} \times B_{ij}) + \sum_{\mathclap{k=1,k \neq i}}^{n} A_{ik} \times B_{kj}                         \\
						                   & = (1 \times B_{ij}) + \sum_{\mathclap{k=1,k \neq i}}^{n} 0 \times B_{kj}\tag{by def. of \(\mathsfbf{1} \)} \\
						                   & = (1 \times B_{ij}) + \sum_{\mathclap{k=1,k \neq i}}^{n} 0\tag{by annihilation prop. of \(0\)}             \\
						                   & = (1 \times B_{ij}) \tag{by identity of \(+\)}                                                             \\
						                   & = B_{ij} \tag{by identity of \(\times \)}
					\end{align*}

					\begin{align*}
						(B \otimes A)_{ij} & = \sum_{k=1}^{n} B_{ik} \times A_{kj}                                                                      \\
						                   & = (B_{ij} \times A_{jj}) + \sum_{\mathclap{k=1,k \neq j}}^{n} B_{ik} \times A_{kj}                         \\
						                   & = (B_{ij} \times 1) + \sum_{\mathclap{k=1,k \neq j}}^{n} B_{ik} \times 0\tag{by def. of \(\mathsfbf{1} \)} \\
						                   & = (B_{ij} \times 1) + \sum_{\mathclap{k=1,k \neq j}}^{n} 0\tag{by annihilation prop. of \(0\)}             \\
						                   & = (B_{ij} \times 1) \tag{by identity of \(+\)}                                                             \\
						                   & = B_{ij} \tag{by identity of \(\times \)}
					\end{align*}

				\item [\(\mathbf{0} \) annihilates \(M\)] \(A \otimes B = \mathbf{0} \) and \(B \otimes A = \mathbf{0} \) where \(A = \mathbf{0} \) iff \(A_{ij} = 0\) for all \(i\), \(j\).

				      \begin{align*}
					      (A \otimes B)_{ij} & = \sum_{k=1}^{n} A_{ik} \times B_{kj}                            \\
					                         & = \sum_{k=1}^{n} 0 \times B_{kj}\tag{by def. of \(\mathbf{0} \)} \\
					                         & = \sum_{k=1}^{n} 0 \tag{by annihilation prop. of \(0\)}          \\
					                         & = 0
				      \end{align*}

				      \begin{align*}
					      (B \otimes A)_{ij} & = \sum_{k=1}^{n} B_{ik} \times A_{kj}                             \\
					                         & = \sum_{k=1}^{n} B_{kj} \times 0 \tag{by def. of \(\mathbf{0} \)} \\
					                         & = \sum_{k=1}^{n} 0 \tag{by annihilation prop. of \(0\)}           \\
					                         & = 0
				      \end{align*}

				\item [Distribution of multiplication over addition]
				\item[Right Distribution]
					\(A \otimes (B \oplus C) = (A \otimes B) \oplus (A \otimes C)\) iff \((A \otimes (B \oplus C))_{ij} = ((A \otimes B) \oplus (A \otimes C))_{ij}\) for all \(i\), \(j\).

					\begin{align*}
						A \otimes (B \oplus C))_{ij} & = \sum_{k=1}^{n} \big( A_{ik} \times (B_{kj} + C_{kj})\big)                                                            \\
						                             & = \sum_{k=1}^{n} \big((A_{ik} \times B_{kj}) + ( A_{ik} \times C_{kj})\big) \tag{by right distribution of \(\times \)} \\
						                             & = \sum_{k=1}^{n} (A_{ik} \times B_{kj}) + \sum_{k=1}^{n} (A_{ik} \times C_{kj})                                        \\
						                             & = (A \otimes B)_{ij} + (A \otimes C)_{ij}                                                                              \\
						                             & = ((A \otimes B) \oplus (A \otimes C))_{ij}
					\end{align*}

				\item[Left Distribution]
					\((A \oplus B) \otimes C = (A \otimes C) \oplus ( B \otimes C)\) iff \(((A \oplus B) \otimes C)_{ij} = ((A \otimes C)\oplus ( B \otimes C))_{ij}\) for all \(i\), \(j\).

					\begin{align*}
						((A \oplus B) \otimes C)_{ij} & = \sum_{k=1}^{n} \big( (A_{ik} + B_{ik}) \times C_{kj} \big)                                                           \\
						                              & = \sum_{k=1}^{n} \big( (A_{ik} \times C_{kj}) + ( B_{ik} \times C_{kj})\big) \tag{by left distribution of \(\times \)} \\
						                              & = \sum_{k=1}^{n} (A_{ik} \times C_{kj}) + \sum_{k=1}^{n} (B_{ik} \times C_{kj})                                        \\
						                              & = (A \otimes C)_{ij} + (B \otimes C)_{ij}                                                                              \\
						                              & = ((A \otimes C) \oplus (B \otimes C))_{ij}
						\qedhere
					\end{align*}
			\end{description}
	\end{description}
\end{proof}

For simplicity, we will write \(\mathbb{M}\) as \(\mathbb{M} (\mathbb{S} ) = (M(S),\mathbf{0} ,\mathsfbf{1} ,\oplus ,\otimes )\).

\subsection{Choices Semi-ring}
\label{sec:app:choice}

This subsection explains and details how functions into semi-ring coefficients can be used to construct semi-rings, and the interplay between this construction and the matrix semi-ring from the previous subsection.

\begin{lemma}
	\label{lem:functions}
	Given a strong semi-ring \(\mathbb{S} = (S, 0, 1, +, \times )\) and a set \(A\), the tuple \(\mathbb{F} = (F, \mathsf{0} , \mathsf{1} , \boxplus , \boxtimes )\), with
	\begin{itemize}
		\item \(F\) the set of functions from \(A\) to \(S\),
		\item \(\mathsf{0} \) the constant function \(\mathsf{0} (a) = 0\) for all \(a \in A\),
		\item \(\mathsf{1} \) the constant function \(\mathsf{1} (a) = 1\) for all \(a \in A\),
		\item \(\boxplus \) defined componentwise: \((f \boxplus g)(a) = (f(a)) + (g(a))\), for all \(f\), \(g\) in \(F\) and \(a \in A\),
		\item \(\boxtimes \) defined componentwise: \((f \boxtimes g)(a) = (f(a)) \times (g(a))\), for all \(f\), \(g\) in \(F\) and \(a \in A\),
	\end{itemize}
	is a strong semi-ring.
\end{lemma}

\begin{proof}
	\begin{description}
		\item[ \((F, \mathsf{0} , \boxplus )\) is a commutative monoid] We first prove that \((F, \mathsf{0} , \boxplus )\) is a commutative monoid by showing that it is associative, commutative, and has \(\mathsf{0} \) as identity.
			\begin{description}
				\item[Associativity]
					\begin{align*}
						((f \boxplus g) \boxplus h)(a) & = (f(a) + g(a)) + h(a)                                          \\
						                               & = f(a) + (g(a) + h(a)) \tag{by assoc. of \(+\)}                 \\
						                               & = (f \boxplus (g \boxplus h))(a) \tag{by def. of \(\boxplus \)}
					\end{align*}

				\item[Commutativity]
					\begin{align*}
						(f \boxplus g)(a) & = f(a) + g(a)                                     \\
						                  & = g(a) + f(a)\tag{by commutativity of \(+\)}      \\
						                  & = (g \boxplus f)(a)\tag{by def. of \(\boxplus \)}
					\end{align*}

				\item[Identity element is \(0\)]
					\begin{align*}
						(\mathsf{0} \boxplus f)(a) & = \mathsf{0} (a) + f(a)                     \\
						                           & = 0 + f(a) \tag{by def. of \(\mathsf{0} \)} \\
						                           & = f(a) \tag{by identity prop of \(+\)}
					\end{align*}
			\end{description}

		\item[\((F, 1, \boxtimes )\) is a monoid] We now prove that \((F, 1, \boxtimes )\) is a monoid by showing that it is associative and has \(\mathsf{1} \) as identity.
			\begin{description}
				\item[Associativity]
					\begin{align*}
						((f \boxtimes g) \boxtimes h)(a) & = (f(a) \times g(a)) \times h(a)                                   \\
						                                 & = f(a) \times (g(a) \times h(a)) \tag{by assoc. of \(\times\) }    \\
						                                 & = (f \boxtimes (g \boxtimes h))(a) \tag{by def. of \(\boxtimes \)}
					\end{align*}

				\item[Identity element is \(1\)]
					\begin{align*}
						(\mathsf{1} \boxtimes f)(a) & = \mathsf{1} (a) \times f(a)                    \\
						                            & = 1 \times f(a)\tag{by def. of \(\mathsf{1} \)} \\
						                            & = f(a) \tag{by identity prop of \(\times\) }
					\end{align*}
			\end{description}

		\item [Distribution of multiplication over addition]
		      We conclude by proving that \(\boxtimes \) distributes over \(\boxplus \).
		      \begin{description}
			      \item[Right Distribution]
				      \begin{align*}
					      (f \boxtimes (g \boxplus h))(a) & = f(a) \times (g(a) + h(a))                                                          \\
					                                      & = (f(a) \times g(a)) + (f(a) \times h(a)) \tag{by right distribution of \(\times \)} \\
					                                      & = ((f \boxtimes g) \boxplus (f \boxtimes h))(a)
				      \end{align*}

			      \item[Left Distribution]
				      \begin{align*}
					      ((f \boxplus g) \boxtimes h)(a) & = (f(a) + g(a)) \times h(a)                                                        \\
					                                      & = (f(a) \times h(a)) + (g(a) \times h(a))\tag{by left distribution of \(\times \)} \\
					                                      & = ((f \boxtimes h) \boxplus (g \boxtimes h))(a)
				      \end{align*}
		      \end{description}

		\item[\(0\) annihilates \(F\)]
			\begin{align*}
				(\mathsf{0} \boxtimes f)(a) & = \mathsf{0} (a) \times f(a)                     \\
				                            & = 0 \times f(a) \tag{by def. of \(\mathsf{0} \)} \\
				                            & = 0 \tag{by annihilation prop of \(0\)}
			\end{align*}

			\begin{align*}
				(f \boxtimes \mathsf{0} )(a) & = f(a) \times \mathsf{0} (a)                     \\
				                             & = f(a) \times 0 \tag{by def. of \(\mathsf{0} \)} \\
				                             & = 0 \tag{by annihilation prop of \(0\)}
			\end{align*}

	\end{description}
\end{proof}

For simplicity, we will write \(\mathbb{F} \) as \(A \to \mathbb{S} = (A \to S, 0, 1, +, \times )\).

\begin{definition}
	\label{def:iso}
	We say two semi-rings \(\mathbb{S} =(S, 0, 1, +, \times )\) and \(\mathbb{T} = (T, \mathsf{0} , \mathsf{1} , \boxplus , \boxtimes )\) are \emph{isomorphic} and write \(\mathbb{S} \cong \mathbb{T} \) if there exists \(g: S \to T\) such that

	\begin{itemize}
		\item \(g\) is a bijection,
		\item \(g(0) = \mathsf{0} \),
		\item \(g(1) = \mathsf{1} \),
		\item \(g(s_1 + s_2) = g(s_1) \boxplus g(s_2)\) for all \(s_1, s_2 \in S\)
		\item \(g(s_1 \times s_2) = g(s_1) \boxtimes g(s_2)\) for all \(s_1, s_2 \in S\)
	\end{itemize}
	For simplicity, we write \(g: \mathbb{S} \to \mathbb{T} \) for such morphisms.
\end{definition}

\begin{lemma}
	\label{lem:semi-ring-iso}
	For all set \(A\) and strong semi-ring \(\mathbb{S} \), \(\mathbb{M} (A \to \mathbb{S} ) \cong A \to \mathbb{M} (\mathbb{S} )\).
\end{lemma}

\begin{proof}
	First, observe that by Lemmas~\ref{lem:matrices} and \ref{lem:functions}, both \(A \to \mathbb{M} (\mathbb{S} )\) and \(\mathbb{M} (A \to \mathbb{S} )\) are strong semi-rings, and we write \(0_f\) (resp.\ \(0_M\)) and \(1_f\) (resp.\ \(1_M\)) for the \(0\) and \(1\) elements of \(A \to \mathbb{M} (\mathbb{S} )\) (resp.\ of \(\mathbb{M} (A \to \mathbb{S} )\)).
	Now we have to prove that we can construct a bijection \(g : M(A \to S) \to (A \to M(S))\) that respects the conditions of \autoref{def:iso}.

	We define \(g\) and \(g^{-1}\) at the same time, then show that they are indeed inverses:
	\begin{description}
		\item[\(g : M(A \to S) \to (A \to M(S))\)]
			Given \(M \in M(A \to S)\) of size \(n \times n\), we let \(g(M) \in A \to M(S)\) be the function that maps \(a \in A\) to \(M\) where the same argument \(a\) has been applied to the functions \(f_{1,1}, \hdots, f_{n,n}\).
			Graphically:

			\[
				g(M)a =
				g(\begin{pmatrix}
					M_{1,1} & \ldots & M_{1,n} \\
					\vdots  & \ddots & \vdots  \\
					M_{n,1} & \ldots & M_{n,n}
				\end{pmatrix})a =
				\begin{pmatrix}
					M_{1,1}a & \ldots & M_{1,n}a \\
					\vdots   & \ddots & \vdots   \\
					M_{n,1}a & \ldots & M_{n,n}a
				\end{pmatrix}
			\]

			Below, we write \(f_M\) for \(g(M)\).

		\item[\(g^{-1}: (A \to M(S)) \to M(A \to S)\)]
			Given \(f \in A \to M(S)\), we define \(g^{-1}(f) \in M(A \to S)\) to be the matrix of size \(n \times n\), for \(n \times n\) the size of the matrix returned by \(f\), such that \((g^{-1}(f))_{i, j}\) is the function that maps \(a \in A\) to \((f(a))_{i,j}\) for all \(i\), \(j\).
			Graphically:

			\[
				g^{-1}(f)a =
				\begin{pmatrix}
					(fa)_{1,1} & \ldots & (fa)_{1,n} \\
					\vdots     & \ddots & \vdots     \\
					(fa)_{n,1} & \ldots & (fa)_{n,n}
				\end{pmatrix}
			\]

			Below, we write \(M_f\) for \(g^{-1}(f)\).
	\end{description}

	\begin{description}
		\item[\(g\) is a bijection]
			We first prove that \(g \circ g^{-1} = g^{-1} \circ g = \id\).
			\begin{description}
				\item[\((g^{-1} \circ g)(M) = M\)]
					\begin{align*}
						(g^{-1} \circ g)(M) & = g^{-1}(g(M))                                                  \\
						                    & = g^{-1}(f_M) \tag{\small{where \((f_M(a))_{ij} = M_{ij}(a)\)}} \\
						                    & = M
					\end{align*}

				\item[\((g \circ g^{-1})(f) = f\)]
					\begin{align*}
						(g \circ g^{-1})(f) & = g(g^{-1}(f))                                             \\
						                    & = g(M_f) \tag{\small{where \((M_f)_{ij}a = (f(a))_{ij}\)}} \\
						                    & = f
					\end{align*}
			\end{description}

		\item[\(g(0_M) = 0_f\)]
			Let \(f = g(0_M)\), then \(f = 0_f\) iff \(f(a)_{ij} = 0_{\mathbb{S}}\) for all \(i\), \(j\).
			\begin{align*}
				f(a)_{ij} & = (0_M)_{ij}(a)                           \\
				          & = 0_f(a)\tag{by def. of \(0_M\)}          \\
				          & = 0_{\mathbb{S}} \tag{by def. of \(0_f\)}
			\end{align*}

		\item[\(g(1_M) = 1_f\)]
			Let \(f = g(1_M)\), then \(f = 1_f\) iff \(f(a)_{ij} = 1_{\mathbb{S}} \) for all \(i=j\) and \(f(a)_{ij} = 0_{\mathbb{S}} \) otherwise.
			\begin{description}
				\item[Case 1: \(i = j\)]
					\begin{align*}
						f(a)_{ij} & = (1_M)_{ij}(a)                           \\
						          & = 1_f(a)\tag{by def. of \(1_M\)}          \\
						          & = 1_{\mathbb{S}} \tag{by def. of \(1_f\)}
					\end{align*}

				\item[Case 2: \(i \neq j\)]
					\begin{align*}
						f(a)_{ij} & = (1_M)_{ij}(a)                           \\
						          & = 0_f(a)\tag{by def. of \(1_M\)}          \\
						          & = 0_{\mathbb{S}} \tag{by def. of \(0_f\)}
					\end{align*}
			\end{description}

		\item[\(g(M_1 + M_2) = g(M_1) + g(M_2)\)]
			\begin{align*}
				     & g(M_1 + M_2) = g(M_1) + g(M_2)                                              \\
				\iff & f_{M_1 + M_2} = f_{M_1} + f_{M_2}                                           \\
				\iff & f_{M_1 + M_2}(a) = (f_{M_1} + f_{M_2})(a)                                   \\
				\iff & f_{M_1 + M_2}(a) = f_{M_1}(a) + f_{M_2} (a)                                 \\
				\iff & (f_{M_1 + M_2}(a))_{ij} = (f_{M_1}(a) + f_{M_2} (a))_{ij}                   \\
				\iff & (M_1 + M_2)_{ij}(a) = (M_1)_{ij}(a)+ (M_2)_{ij}(a) \tag{by assoc. of \(+\)}
			\end{align*}

		\item[\(g(M_1 \times M_2) = g(M_1) \times g(M_2)\)]
			\begin{align*}
				     & g(M_1 \times M_2) = g(M_1) \times g(M_2)                                                                                                      \\
				\iff & f_{M_1 \times M_2} = f_{M_1} \times f_{M_2}                                                                                                   \\
				\iff & f_{M_1 \times M_2}(a) = (f_{M_1} \times f_{M_2})(a)                                                                                           \\
				\iff & f_{M_1 \times M_2}(a) = (f_{M_1})(a) \times (f_{M_2})(a)                                                                                      \\
				\iff & (f_{M_1 \times M_2}(a))_{ij} = ((f_{M_1})(a) \times (f_{M_2})(a))_{ij}                                                                        \\
				\iff & (\sum_{k=1}^{n} (M_1)_{ik} \times (M_2)_{kj})(a) = \sum_{k=1}^{n} (M_1)_{ik}(a) \times (M_2)_{kj}(a) \tag{by assoc. of \(+\) and \(\times \)}
			\end{align*}
			\qedhere
	\end{description}
\end{proof}

\subsection{Partiality}
\label{sec:app:partiality}

In our improvement of the analysis, we add an \(\infty\) element to the mwp-semi-ring, but reason abstractly below with an arbitrary semi-ring and a \(\bot\) element.

\begin{lemma}
	Given a strong semi-ring \(\mathbb{S} = (S,0,1,+,\times )\) and an element \(\bot \notin S\), \(\mathbb{S}^{\bot} = (S \cup \{\bot\},0,1,+^{\bot},\times^{\bot} )\) with, for all \(a\), \(b \in S \cup \{\bot\}\),
	\begin{align*}
		a +^{\bot} b      & =\begin{dcases*}
			a + b & if \(a, b \neq \bot\) \\
			\bot  & otherwise
		\end{dcases*} \\
		a \times^{\bot} b & =\begin{dcases*}
			a \times b & if \(a, b \neq \bot\) \\
			\bot       & otherwise
		\end{dcases*}
	\end{align*}
	is a semi-ring.
\end{lemma}

\begin{proof}
	The proof is immediate, but note that \(\mathbb{S}^{\bot}\) is not strong, as \(\bot \times 0 = \bot\).
\end{proof}

A good intuition on this construction comes from partial functions.
Indeed, we can define \(A \rightharpoonup \mathbb{S}\) as the semi-ring of partial functions from \(A\) to \(\mathbb{S}\), i.e.\ of functions from \(A\) to \(\mathbb{S}^{\bot}\).
Furthermore, if we identify a matrix in \(\mathbb{M} (\mathbb{S}^{\bot})\) where at least a coefficient is \(\bot \) with the matrix \(\bot\), then we get that \(\mathbb{M} (A \rightharpoonup \mathbb{S} ) \cong A \rightharpoonup \mathbb{M} (\mathbb{S} )\).
However, note that none of those semi-rings are strong.

\end{document}